\documentclass[letterpaper, 10 pt, conference]{ieeeconf}  

\IEEEoverridecommandlockouts                              
\usepackage[draft,inline,nomargin]{fixme}   
\usepackage{graphicx}

\usepackage{amsmath}
\usepackage{dsfont}                         
\newcommand{\tn}[1]{\textnormal{#1}}        

\usepackage{amssymb}  

\newtheorem{definition}{Definition}
\newtheorem{proposition}{Proposition}
\newtheorem{problem}{Problem}

\title{\LARGE \bf
Proofs for an Abstraction of Continuous Dynamical Systems Utilizing Lyapunov Functions
}

\author{Christoffer Sloth and Rafael Wisniewski
\thanks{This work was supported by MT-LAB, a VKR Centre of Excellence.}%
\thanks{Christoffer Sloth is with Department of Computer Science,
        Aalborg University, 9220 Aalborg East, Denmark
        {\tt\small csloth@cs.aau.dk}}%
\thanks{Rafael Wisniewski is with the Section of Automation \& Control, Aalborg University,
        9220 Aalborg East, Denmark
        {\tt\small raf@es.aau.dk}}%
}

\begin{document}

\maketitle
\thispagestyle{empty}
\pagestyle{empty}

\begin{abstract}
In this report proofs are presented for a method for abstracting continuous dynamical systems by timed automata. The method is based on partitioning the state space of dynamical systems with invariant sets, which form cells representing locations of the timed automata.

To enable verification of the dynamical system based on the abstraction, conditions for obtaining sound, complete, and refinable abstractions are set up.

It is proposed to partition the state space utilizing sub-level sets of Lyapunov functions, since they are positive invariant sets. The existence of sound abstractions for Morse-Smale systems and complete and refinable abstractions for linear systems are proved.
\end{abstract}

\section{Introduction}\label{sec:introduction}
Verifying properties such as safety is important for any system. Such verification is based on reachability calculations or approximations. Since the exact reachable sets of continuous and hybrid systems in general are incomputable \cite{recent_progress} a lot of attention has been paid to their approximations. Yet reachability is decidable for discrete systems such as automata and timed automata; consequently, there exists a rich set of tools aimed at verifying properties of such systems. Therefore, abstracting dynamical systems by discrete systems would enable verification of dynamical systems using these tools.

There are basically two methods for verifying continuous and hybrid systems. The first is to over-approximate the reachable states by simple convex sets as in \cite{ellipsoidal_book}. The second method is based on abstracting the original system into a description with reduced complexity, while preserving certain properties of the original systems. This is accomplished for hybrid systems in \cite{abstractions_for_hybrid_systems} and for continuous systems in \cite{oded}.

In this work, continuous systems are abstracted by timed automata. This concept is primarily motivated by \cite{oded} where slices are introduced to improve abstractions of continuous systems. A slice is a counterpart of a single direction in continuous systems.

This technical report is devoted to proving the propositions presented in the paper "Abstraction of Continuous Dynamical Systems Utilizing Lyapunov Functions", written by Christoffer Sloth and Rafael Wisniewski for the 49$^{\tn{th}}$ IEEE Conference on Decision and Control (CDC) \cite{CDC2010}. Therefore, that paper can be consulted for further insight in the abstraction method.
In that paper the idea of considering both cells and slices for abstractions was adopted to provide as solution to the following problem.

\begin{problem}\label{problem:intro}
Given an autonomous dynamical system, find a partition of its state space, which allows arbitrary close over-approximation of the reachable set by a timed automaton.
\end{problem}

The abstraction to be addressed preserves safety and has an upper bound on the size of the over-approximation of the reachable set. Furthermore, it is possible to reduce the size of the upper bound to an arbitrary small value, for a class of systems, by refining the partitioning. Hence, we can obtain an abstraction with arbitrary precision of the reachable set.

\section{Preliminaries}\label{sec:preliminaries}
The purpose of this section is to provide some definitions related to autonomous dynamical systems and timed automata.

An autonomous dynamical system $\Gamma=(X,f)$ is a system with state space $X\subseteq\mathds{R}^{n}$ and dynamics described by ordinary differential equations $f:X\rightarrow\mathds{R}^{n}$
\begin{align}
\dot{x}&=f(x).\label{eqn:auto_differential_equation}
\end{align}
The function $f$ is assumed to be locally Lipschitz. Additionally, we assume linear growth of $f$, then according to Theorem~1.1 in \cite{nonsmooth_analysis_and_control_theory} there exists a solution of \eqref{eqn:auto_differential_equation} on $(-\infty,\infty)$.
%

The solution of \eqref{eqn:auto_differential_equation}, from an initial state $x_{0}\in X$ at time $t\geq0$ is described by the flow function $\phi_{\Gamma}:[0,\epsilon]\times X\rightarrow X$, $\epsilon>0$ satisfying
\begin{align}
\frac{d\phi_{\Gamma}(t,x_{0})}{dt}&=f\left(\phi_{\Gamma}(t,x_{0})\right)\label{eqn:solution_of_auto_differential_equation}
\end{align}
for all $t\geq0$.

Lyapunov functions are utilized in stability theory and are defined in the following \cite{Energy_Functions_for_Morse_Smale_Systems}.
\begin{definition}[Lyapunov Function]\label{def:lyapunov_function}
Assume that a mapping $f:\mathds{R}^{n}\rightarrow\mathds{R}^{n}$ is continuous on $G\subset\mathds{R}^{n}$ and that $G$ is open and connected. Then a real non-degenerate function $\psi:\mathds{R}^{n}\rightarrow\mathds{R}\cup\{-\infty,\infty\}$ differentiable on $G$ is said to be a Lyapunov function for the differential equation shown in \eqref{eqn:auto_differential_equation} if
\begin{align}
\notag&p\tn{ is a critical point of }f\Leftrightarrow p\tn{ is a critical point of }\psi\\
&\dot{\psi}(x)=\sum_{j=1}^{n}\frac{\partial \psi}{\partial x_{j}}(x)f^{j}(x)\label{eqn:Lyap_der}\\
&\dot{\psi}(x):
\begin{cases}
=0\indent\tn{if }x=p\\
<0\indent\tn{if }x\in G\backslash\{p\}
\end{cases}
\end{align}
and $\exists\tn{ }\alpha>0$ and an open neighborhood of each critical point $p$, where
\begin{align}
&||\dot{\psi}(x)||\geq\alpha||x-p||.
\end{align}
\end{definition}
Note that we do not require positive definiteness of $\psi$.

\begin{definition}[Reachability for Dynamical System]
The reachable set of a dynamical system $\Gamma$ from a set of initial states $X_{0}\subseteq X$ on the time interval $[t_{1},t_{2}]$ is defined as
\begin{align}
\notag\tn{Reach}_{[t_{1},t_{2}]}(\Gamma,X_{0})=\{&x\in X|\exists t\in[t_{1},t_{2}]\tn{, }\exists x_{0}\in X_{0},\\
&\tn{such that }x =\phi_{\Gamma}(t,x_{0})\}.
\end{align}
\end{definition}

The dynamical system will be abstracted by a timed automaton. Therefore, a definition of timed automaton is provided in the following \cite{A_theory_of_timed_automata}.
In the definition, a set of clock constraints $\Psi(C)$ for the set $C$ of clocks is utilized. $\Psi(C)$ contains all invariants and guards of the timed automaton, consequently it is described by the following grammar \cite{Quantitative_Analysis_of_Weighted_Transition_Systems}:
\begin{align}
&\psi::=c_{1}\bowtie k|c_{1}-c_{2}\bowtie k|\psi_{1}\wedge\psi_{2}\tn{, where}\label{eqn:clock_grammar}\\
&\notag c_{1},c_{2}\in C,\, k\in\mathds{Z},\tn{ and }\bowtie\in\{\leq,<,=,>,\geq\}.
\end{align}
Note that the clock constraint $k$ should be an integer, but in this paper no effort is done in converting the clock constraints into integers.
\begin{definition}[Timed Automaton]
A timed automaton, $\mathcal{A}$, is a tuple $(L, L_{0}, C, \Sigma, I, \Delta)$, where
\begin{itemize}
\item $L$ is a finite set of locations, and $L_{0}\subseteq L$ is the set of initial locations.
\item $C$ is a finite set of clocks all with values in $\mathds{R}_{\geq0}$.
\item $\Sigma$ is the input alphabet.
\item $I:L\rightarrow\Psi(C)$ assigns invariants to locations, where $\Psi(C)$ is the set of all clock constraints, see \eqref{eqn:clock_grammar}.
\item $\Delta\subseteq L\times\Psi(C)\times\Sigma\times2^{C}\times L$ is a finite set of transition relations. The transition relations provide edges between locations as tuples $(l,G_{l\rightarrow l'},\sigma,R_{l\rightarrow l'},l')$, where $l$ is the source location, $l'$ is the destination location, $G_{l\rightarrow l'}\in\Psi(C)$ is the guard set, $\sigma$ is a symbol in the alphabet $\Sigma$, and $R_{l\rightarrow l'}\in2^{C}$ gives the set of clocks to be reset.
\end{itemize}
\end{definition}
We use the mapping $v:C\rightarrow\mathds{R}_{\geq0}$ for a clock valuation on a set of clocks $C$. Additionally, the initial valuation is denoted $v_{0}$, where $v_{0}(c)=0$ for all $c\in C$.

Analog to the solution of \eqref{eqn:auto_differential_equation} shown in \eqref{eqn:solution_of_auto_differential_equation}, a run of a timed automaton is defined in the following.
\begin{definition}[Run of Timed Automaton]\label{def:run_of_timed_automaton}
A run of a timed automaton $\mathcal{A}$ is a possibly infinite sequence of alternations between time steps and discrete steps in the following form
\begin{align}
(v_{0},l_{0})\overset{t_{1}}{\longrightarrow}(v_{0}+t_{1},l_{0})\overset{\sigma_{1}}{\longrightarrow}(v_{1},l_{1})\longrightarrow\dots
\end{align}
The multifunction describing a run of a timed automaton is $\phi_{\mathcal{A}}:\mathds{R}_{\geq0}\times L_{0}\rightarrow 2^{L}$. Here $l\in\phi_{\mathcal{A}}(t,l_{0})$ if and only if the timed automaton $\mathcal{A}$ initialized in $l_{0}$ can be in location $l$ at time $t=\sum_{i} t_{i}$.
\end{definition}

From the run of a timed automaton, the reachable set is defined below.
\begin{definition}[Reachability for Timed Automaton]
The reachable set of a timed automaton $\mathcal{A}$ with initial locations $L_{0}$ on the time interval $[t_{1},t_{2}]$ is defined as
\begin{align}
\notag\tn{Reach}_{[t_{1},t_{2}]}(\mathcal{A},L_{0})=\{&l\in L|\exists t\in[t_{1},t_{2}],\exists l_{0}\in L_{0},\\
&\tn{such that }l \in\phi_{\mathcal{A}}(t,l_{0})\}.
\end{align}
\end{definition}

\section{Generation of Finite Partition}\label{sec:generation_of_finite_partition}
A finite partition of the state space of the considered system is generated using slices, which are set-differences between positive invariant sets.
\begin{proposition}
If $S_{1}\pitchfork S_{2}\neq\emptyset$ then
\begin{align}
&\tn{int}(S_{1}\cap S_{2})\neq\emptyset.
\end{align}
\end{proposition}
\begin{proof}
Let $p\in \tn{bd}(S_{1})\pitchfork \tn{bd}(S_{2})$ by Theorem~7.7 in \cite{topology_and_geometry} there exists a local coordinate system $(\Upsilon,U)$ such that
\begin{subequations}
\begin{align}
&\Upsilon(S_{1}\cap U)=H_{1}^{+}\subset\mathds{R}^{n}\\
&\Upsilon(S_{2}\cap U)=H_{2}^{+}\subset\mathds{R}^{n}
\end{align}
\end{subequations}
where $H_{1}^{+}$ and $H_{2}^{+}$ are supporting hyperplanes of $S_{1}$ and $S_{2}$. Thus $\tn{dim}(H_{1}^{+}\pitchfork H_{2}^{+})=n$.
\end{proof}

Note that the intersection of slices may form multiple disjoint sets. Therefore, the intersection of $k$ slices is denoted an extended cell $e_{\tn{ex},g}$. Each of the disjoint sets of an extended cell $e_{\tn{ex},g}$ is called a cell $e_{g,h}$.

\section{Generation of Timed Automaton from Finite Partition}\label{sec:obtaining_TA}
A timed automaton is generated by associating each cell of a partition with a location and by inserting guards and invariants calculated based on the dynamics. The method is presented in \cite{CDC2010} and is very similar to the method presented in \cite{oded}.
\begin{proposition}\label{prop:deterministic_TA}
$\mathcal{A}\left(\mathcal{S}\right)$ is a deterministic timed automaton, if and only if for each cell $e_{(g,h)}$ and for all $i=1,\dots,k$ the set
\begin{align}\label{eqn:deterministic_TA}
e_{(g,h)}\bigcap \psi_{i}^{-1}(a_{(i,g_{i}-1)})
\end{align}
is connected.
\end{proposition}
\begin{proof}
If $e_{(g,h)}\bigcap \psi_{i}^{-1}(a_{(i,g_{i}-1)})$ is not connected for some $i$, then $\sigma_{i}$ is the label of multiple outgoing transitions from the location $e_{(g,h)}$, i.e. there exist multiple transitions in $\Delta$, where $e_{(g,h)}$ is the source location and $\sigma_{i}$ is the label. Therefore, the timed automaton $\mathcal{A}\left(\mathcal{S}\right)$ is nondeterministic.
\end{proof}

\begin{proposition}\label{prop:parallel_comp}
Let $\mathcal{A}_{\tn{ex}}(\mathcal{S})$ be a timed automaton, with locations associated to extended cells, and let the slices of $\mathcal{S}$ be generated such that for each pair $S_{(i,g_{i})}$ and $S_{(j,g_{j})}$, with $i,j\in\{1,\dots,k\}$, $g_{i}\in\{1,\dots,|\mathcal{S}_{i}|\}$, $g_{j}\in\{1,\dots,|\mathcal{S}_{j}|\}$, we have
\begin{align}
S_{(i,g_{i})}\pitchfork S_{(j,g_{j})}\neq\emptyset\indent\forall i\neq j.\label{eqn:parallel_comp}
\end{align}
Then $\mathcal{A}_{\tn{ex}}(\mathcal{S})$ is isomorphic to the parallel composition of $k$ timed automata each generated by one slice-family $\mathcal{S}_{i}$.
\end{proposition}
\begin{proof}
Consider the timed automaton $\mathcal{A}_{||}(\mathcal{S})=\mathcal{A}_{1}(\mathcal{S}_{1})||\dots||\mathcal{A}_{k}(\mathcal{S}_{k})$ where $\mathcal{A}_{i}(\mathcal{S}_{i})=(L_{i},L_{0,i},C_{i},\Sigma_{i},I_{i},\Delta_{i})$ and $L_{i}=\{l_{(i,1)},\dots,l_{(i,|\mathcal{S}_{i}|)}\}$, abstracting the slices $S_{(i,1)},\dots,S_{(i,|\mathcal{S}_{i}|)}$. Then the timed automaton $\mathcal{A}_{||}(\mathcal{S})$ is given by
\begin{itemize}
\item \textbf{Locations:} $L=L_{1}\times\dots\times L_{k}$, which according to Definition~10 in \cite{CDC2010} represents extended cells, if the transversal intersection of all slices is nonempty i.e. \eqref{eqn:parallel_comp} is satisfied.
\item \textbf{Clocks:} $C=\{c_{i},\dots,c_{k}\}$, where $c_{i}$ monitors the time for being in a slice of $\mathcal{S}_{i}$.
\item \textbf{Invariants:} The invariant for location $l_{\tn{ex},g}=(l_{(1,g_{1})},\dots,l_{(k,g_{k})})$ is identical to (18) in \cite{CDC2010} and is
\begin{align}
I(l_{\tn{ex},g})&= \bigwedge_{i=1}^{k}I_{i}(l_{(i,g_{i})}).
\end{align}
\item \textbf{Input Alphabet:} $\Sigma=\{\sigma_{1},\dots,\sigma_{k}\}$.
\item \textbf{Transition relations:} $\Sigma_{i}$ is disjoint from $\Sigma_{j}$ for all $i\neq j$; hence, item 1) in Definition~15 in \cite{CDC2010} never happens. 
\end{itemize}
This implies that $\mathcal{A}_{||}(\mathcal{S})=\mathcal{A}_{1}(\mathcal{S}_{1})||\dots||\mathcal{A}_{k}(\mathcal{S}_{k})$ and $\mathcal{A}_{\tn{ex}}(\mathcal{S})$ are isomorph.
\end{proof}
\setcounter{proposition}{4}
\begin{proposition}
Let $\mathcal{S}=\{\mathcal{S}_{1},\dots,\mathcal{S}_{k}\}$ be a collection of slice-families, and $\psi_{i}$ be a partitioning function for $\mathcal{S}_{i}$. The timed automata $\mathcal{A}_{\tn{ex}}(\mathcal{S})$ and $\mathcal{A}(\mathcal{S})$ are bisimilar if for each cell $e_{(g,h)}\in K(\mathcal{S})$ and each $i\in\{1,\dots,k\}$
\begin{subequations}\label{eqn:prop_eq_bisimilar}
\begin{align}
&e_{(g,h)}\bigcap \psi_{i}^{-1}(a_{(i,g_{i}-1)})\neq\emptyset\indent \forall\,h\tn{ or}\\
&e_{(g,h)}\bigcap \psi_{i}^{-1}(a_{(i,g_{i}-1)})=\emptyset\indent\forall\,h.
\end{align}
\end{subequations}
\end{proposition}
If \eqref{eqn:prop_eq_bisimilar} holds, then all cells in each extended cell have the same symbols on their outgoing transitions.

\begin{proof}
Let $e_{(g,h)}$ with $h=1,\dots,m$ be the cells which union is the extended cell $e_{\tn{ex},g}$. Then
\begin{align}
I(e_{(g,h)})=I(e_{(g,k)})\indent\forall h,k\in\{1,\dots,m\}
\end{align}
as the invariants are calculated based on slices (18) in \cite{CDC2010}.

If the partition satisfies \eqref{eqn:prop_eq_bisimilar}, then the same outgoing transitions exist for all cells within the same extended cell. Furthermore,
\begin{align}
G_{(g,h)\rightarrow(g',h')}=G_{(g,k)\rightarrow(g',k')}\indent\forall h,k\in\{1,\dots,m\}
\end{align}
since the guards are also calculated based on slices (19b) in \cite{CDC2010}. This implies that all possible behaviors from each cell in an extended cell are the same; hence, $\mathcal{A}(\mathcal{S})$ is bisimilar to a timed automaton $\mathcal{A}_{\tn{ex}}(\mathcal{S})$.
\end{proof}

\section{Conditions for the Partitioning}\label{sec:partitioning_ss}
A sound and a complete abstraction of a dynamical system is illustrated in Fig.~\ref{fig:tight_definition}.
\begin{figure}[!htb]
    \centering
       \includegraphics[scale=1]{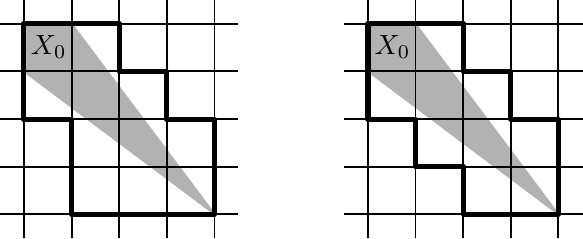}
       \caption{Illustration of the reachable set of a dynamical system (gray) from initial set $X_{0}$ and a sound approximation of this (cells within bold black lines) on the left and a complete abstraction on the right.}
       \label{fig:tight_definition}
\end{figure}
Definitions of sound and complete abstractions are available in \cite{CDC2010}.

\begin{proposition}\label{prop:par_also_sound_complete}
A timed automaton $\mathcal{A}_{\tn{ex}}=\mathcal{A}_{1}||\dots||\mathcal{A}_{k}$, with locations abstracting extended cells, is a sound (complete) abstraction of the system $\Gamma$ if and only if $\mathcal{A}_{1},\dots,\mathcal{A}_{k}$ are sound (complete) abstractions of $\Gamma$.
\end{proposition}
\begin{proof}
If the locations of $\mathcal{A}_{\tn{ex}}$ are extended cells, then soundness of $\mathcal{A}_{\tn{ex}}$ can be reformulated to the following.


A timed automaton $\mathcal{A}_{\tn{ex}}$ with $L_{0}=\{e_{\tn{ex},g}|g\in\mathcal{G}_{0}\subseteq\mathcal{G}\}$ is said to be a sound abstraction of $\Gamma$ with $X_{0}=\bigcup_{g\in\mathcal{G}_{0}}e_{\tn{ex},g}$ on $[t_{1}, t_{2}]$ if for all $t\in[ t_{1}, t_{2}]$ and for all $g\in\mathcal{G}$
\begin{subequations}
\begin{align}
&\bigcap_{i=1}^{k}S_{(i,g_{i})}\cap\tn{Reach}_{[t,t]}(\Gamma,X_{0})\neq\emptyset\indent\tn{implies}\\
&\notag\exists l_{0}\in L_{0}\tn{ such that}\\
&\bigcap_{i=1}^{k}S_{(i,g_{i})}\in\alpha_{K}^{-1}(\phi_{\mathcal{A}_{\tn{ex}}}(t,l_{0}))
\end{align}
\end{subequations}
which is equivalent to: For all $i=\{1,\dots,k\}$, all $g\in\mathcal{G}$, and for all $t\in[ t_{1}, t_{2}]$
\begin{subequations}\label{eqn:sound_c}
\begin{align}
&S_{(i,g_{i})}\cap\tn{Reach}_{[t,t]}(\Gamma,X_{0})\neq\emptyset\indent\tn{implies}\\
&\notag\exists l_{0,i}\in L_{0,i}\tn{ such that}\\
&\alpha_{K}^{-1}(\phi_{\mathcal{A}_{i}}(t,l_{0,i}))= S_{(i,g_{i})}.
\end{align}
\end{subequations}
From \eqref{eqn:sound_c} it is seen that $\mathcal{A}_{\tn{ex}}=\mathcal{A}_{1}||\dots||\mathcal{A}_{k}$ is sound if and only if $\mathcal{A}_{i}$ is sound for $i=1,\dots,k$. Similar arguments can be used to prove the completeness part of Proposition~\ref{prop:par_also_sound_complete}.
\end{proof}

\begin{proposition}[Sufficient Condition for Soundness]\label{prop:suf_cond_soundness}
A timed automaton $\mathcal{A}(\mathcal{S})$ is a sound abstraction of the system $\Gamma$, if its invariants and guards are formed using
\begin{subequations}\label{eqn:suf_cond_soundness}
\begin{align}
\underline{t}_{S_{(i,g_{i})}}&\leq\frac{|a_{(i,g_{i})}-a_{(i,g_{i}-1)}|}{\sup\{|\dot{\psi}_{i}(x)|\in\mathds{R}_{\geq0}|x\in S_{(i,g_{i})}\}}\\
\overline{t}_{S_{(i,g_{i})}}&\geq\frac{|a_{(i,g_{i})}-a_{(i,g_{i}-1)}|}{\inf\{|\dot{\psi}_{i}(x)|\in\mathds{R}_{\geq0}|x\in S_{(i,g_{i})}\}}
\end{align}
\end{subequations}
where $\dot{\psi}_{i}(x)$ is defined as shown in \eqref{eqn:Lyap_der}.
\end{proposition}
\begin{proof}
Let $\mathcal{A}(\mathcal{S})$ be a timed automaton with $L_{0}=\{e_{i}|i\in\mathcal{I}\}$, be an abstraction of $\Gamma$ with initial set $X_{0}=\bigcup_{i\in\mathcal{I}}e_{i}$. If guards and invariants of $\mathcal{A}(\mathcal{S})$ satisfy \eqref{eqn:suf_cond_soundness}, then
\begin{align}
\tn{Reach}_{[t_{1},t_{2}]}(\Gamma,X_{0})\subseteq\alpha_{K}^{-1}(\tn{Reach}_{[t_{1},t_{2}]}(\mathcal{A},L_{0}))
\end{align}
since for all $x_{0}\in \psi_{i}^{-1}(a_{(i,g_{i})})$ there exists $t\in[\underline{t}_{S_{(i,g_{i})}},\overline{t}_{S_{(i,g_{i})}}]$ such that
\begin{align}
\phi_{\Gamma}(t,x_{0})\in \psi_{i}^{-1}(a_{(i,g_{i}-1)}).
\end{align}
\end{proof}
\begin{proposition}[Sufficient Condition for Completeness]\label{prop:suf_cond_completeness}
Let $\mathcal{S}=\{\mathcal{S}_{i}|i=1,\dots,k\}$ be a collection of slice-families and let
\begin{align}
S_{(i,g_{i})}=\psi_{i}^{-1}([a_{(i,g_{i}-1)}, a_{(i,g_{i})}]).
\end{align}
A deterministic timed automaton is a complete abstraction if
\begin{enumerate}
\item $\overline{t}_{S_{(i,g_{i})}}=\underline{t}_{S_{(i,g_{i})}}=t_{(i,g_{i})}$ and
\item for any $g\in\mathcal{G}$ with $g_{i}\geq2$ there exists a time $t_{(i,g_{i})}$ such that $\forall$ $x_{0}\in \psi_{i}^{-1}(a_{(i,g_{i})})$
\begin{align}
\phi_{\Gamma}(t_{(i,g_{i})},x_{0})\in \psi_{i}^{-1}(a_{(i,g_{i}-1)}).
\end{align}
\end{enumerate}
\end{proposition}
\begin{proof}
The proposition states that it takes the same time for all trajectories of $\Gamma$ to propagate between any two level sets of $\psi_{i}$. From this it follows that $\mathcal{A}(\mathcal{S})$ is complete if $\overline{t}_{S_{(i,g_{i})}}$ and $\underline{t}_{S_{(i,g_{i})}}$ are equal to $t_{(i,g_{i})}$.
\end{proof}
\begin{proposition}[Nec. Cond. for Refinable Abstraction]\label{nec:refinable_abstraction}
If $\mathcal{A}(\mathcal{S})$ is a refinable abstraction of a system $\Gamma$, then $\mathcal{S}$ is a collection of $n$ slice-families.
\end{proposition}
\begin{proof}
If $\mathcal{A}(\mathcal{S})$ is a refinable abstraction, then for any $\epsilon>0$ there exists a partitioning $K(\mathcal{S})$ such that (30) in \cite{CDC2010} holds for cells in $K(\mathcal{S})$. Therefore,
\begin{align}
S_{(i,g_{i})}\subset \psi_{i}^{-1}(a_{(i,g_{i})})+B(\epsilon)
\end{align}
where $\epsilon>0$. Note that $a_{(i,g_{i})}$ is a regular value of $\psi_{i}$, i.e. the dimension of the level set $\psi_{i}^{-1}(a_{(i,g_{i})})$ is $n-1$.
The locations of $\mathcal{A}(\mathcal{S})$ are cells for which
\begin{subequations}
\begin{align}
\bigcup_{h} e_{(g,h)}&=\pitchfork_{i=1}^{k}S_{(i,g_{i})}\\
&\subset\pitchfork_{i=1}^{k} \left(\psi_{i}^{-1}(a_{(i,g_{i})})+B(\epsilon)\right)\\
&\subset\pitchfork_{i=1}^{k} \psi_{i}^{-1}(a_{(i,g_{i})})+B(2\epsilon).\label{eqn:eex_dim}
\end{align}
\end{subequations}
But \eqref{eqn:eex_dim} is true for any $\epsilon$, thus it is enough to prove that
\begin{align}
\tn{dim}\left(\pitchfork_{i=1}^{k} \psi_{i}^{-1}(a_{(i,g_{i})})\right)=0.
\end{align}
Using Theorem~7.7 in \cite{topology_and_geometry} the dimension of an extended cell is given by
\begin{subequations}
\begin{align}
\notag\tn{dim}&\left(\pitchfork_{i=1}^{k} \psi_{i}^{-1}(a_{(i,g_{i})})\right)\\\notag&=\left[\left[(n-1)+(n-1)-n\right]+(n-1)-n\right]\\&+(n-1)-n\dots\\
&=k(n-1)-(k-1)n.\label{eqn:dim_ext_cell}
\end{align}
\end{subequations}
We see that if $k\neq n$ then $\tn{dim}\left(\pitchfork_{i=1}^{k} \psi_{i}^{-1}(a_{(i,g_{i})})\right)\neq0$, thus we have contradiction. We conclude that $k=n$.
\end{proof}

\section{Partitioning the State Space Using Lyapunov Functions}\label{sec:partitioning_Lyapunov}
Positive invariant sets are used in stability theory in the form of sub-level sets of Lyapunov functions. This concept is adopted in this work to synthesize partitions.

\begin{definition}
Two Lyapunov functions $\psi_{1},\psi_{2}:\mathds{R}^{n}\rightarrow\mathds{R}$ are transversal if the level sets $\psi_{1}^{-1}(a)$ and $\psi_{2}^{-1}(a)$ are transversal for any $a\in\mathds{R}\backslash\{0\}$.
\end{definition}
\begin{proposition}
Let $n>1$. For any Morse-Smale system (see Chapter~4 in \cite{Geometric_Theory_of_Dynamical_Systems_An_Introduction}) on $\mathds{R}^{n}$ there exists $n$ transversal Lyapunov functions.
\end{proposition}
\begin{proof}
Let $S(n,\mathds{R})$ be a set of $n\times n$ symmetric matrices. $S(n,\mathds{R})$ is a subspace of $M(n,\mathds{R})$ of $\tn{dim}\left(S(n,\mathds{R})\right)=n(n+1)/2$. Consider the map $\psi_{A}:S(n,\mathds{R})\rightarrow S(n,\mathds{R})$ and let
\begin{align}
P\mapsto A^{\tn{T}}P+PA.
\end{align}
Now consider the map $\tn{det}:M(n,\mathds{R})\rightarrow \mathds{R}$ and let
\begin{align}
A\mapsto \tn{det}(A).
\end{align}
Then $(\tn{det}\circ\psi_{A})^{-1}(\{0\})$ is a closed set. Therefore,
\begin{align}
U_{A}\equiv\{P\in S(n,\mathds{R})|\tn{det}\circ\psi_{A}(P)\neq0\}
\end{align}
is an open set.
$V_{A}\equiv V\cap U_{A}$ is open, where
\begin{align}
V=\{P\in S(n,\mathds{R})|\tn{det}(P)\neq0\}\indent(V\tn{ is open}).
\end{align}
Let $\Theta=\{Q\in S(n,\mathds{R})|Q>0\}$ by Proposition~2.18 in \cite{Geometric_Theory_of_Dynamical_Systems_An_Introduction} the map
\begin{align}
&M(n,\mathds{R})\rightarrow C^{n}/S^{n}\tn{ defined by}\\
&L\mapsto \tn{diag}([\lambda_{1},\dots,\lambda_{n}])\tn{ is continious.}
\end{align}
Thus $\Theta$ is an open set in $S(n,\mathds{R})$.\\
We pick an open neighborhood around $Q=A^{\tn{T}}P+PA$ and denote it $U$. Then for every $Q'\in U$ there exists a (unique) $P$, thus $\psi_{A}^{-1}(U)$ is a nonempty open set in $S(n,\mathds{R})$.

We can pick $n$ linear independent matrices $P_{1},\dots,P_{n}\in\psi_{A}^{-1}(U)$. This is possible because $\psi_{A}^{-1}(U)$  is open in $S(n,\mathds{R})$ and $\tn{dim}(S(n,\mathds{R}))$ is $n(n+1)/2$. Then for any $a\in\mathds{R}\backslash\{0\}$ and $i\neq j$
\begin{align}
\{x\in\mathds{R}^{n}|x^{\tn{T}}P_{i}x=a\}\pitchfork\{x\in\mathds{R}^{n}|x^{\tn{T}}P_{j}x=a\}.
\end{align}
Extending this to Morse-Smale systems follows directly from Theorem~1 in \cite{Energy_Functions_for_Morse_Smale_Systems}.
\end{proof}

\subsection{Complete Abstraction}
A complete abstraction of \eqref{eqn:auto_differential_equation} can be obtained by constructing a partition generated by Lyapunov functions, which satisfies Proposition~\ref{prop:suf_cond_completeness}.

\begin{proposition}\label{prop:ly_tight}
Let each slice-family of $\mathcal{S}=\{\mathcal{S}_{i}|i=1,\dots,k\}$ be associated with a Lyapunov function $\psi_{i}(x)$ for the system $\Gamma$, such that
$S_{(i,j)}=\psi_{i}^{-1}([a_{(i,j-1)},a_{(i,j)}])$
and let
\begin{align}
\psi_{i}(x)=\alpha\dot{\psi}_{i}(x)\indent\forall x\in\mathds{R}^{n}.
\end{align}
Then $\mathcal{A}(\mathcal{S})$ is a complete abstraction of $\Gamma$.
\end{proposition}
\begin{proof}
Let $\psi(x)$ be a Lyapunov function for the system $\Gamma$ and let $x,x'\in\psi^{-1}(a_{m})$. According to Proposition~\ref{prop:suf_cond_completeness} the abstraction is complete if there exists a $t_{m}$, for $m=2,\dots,k$ such that
\begin{align}
\phi_{\Gamma}(t_{m},x),\phi_{\Gamma}(t_{m},x')\in\psi^{-1}(a_{m-1}).\label{eqn:lf_proof1}
\end{align}
This is true if
\begin{align}
\dot{\psi}(\phi_{\Gamma}(t,x))-\dot{\psi}(\phi_{\Gamma}(t,x'))=0\indent\forall t.\label{eqn:lf_proof2}
\end{align}
The combination of \eqref{eqn:lf_proof1} and \eqref{eqn:lf_proof2} implies that for all $c>0$ there exists an $\alpha$ such that
\begin{align}
\psi^{-1}(c)=\dot{\psi}^{-1}\left(\frac{c}{\alpha}\right)
\end{align}
hence for all $x$ there exists an $\alpha$ such that
\begin{align}
\psi(x)=\alpha\dot{\psi}(x).
\end{align}
\end{proof}
\begin{proposition}
For any hyperbolic linear system $\Gamma$ there exists $n$ transversal Lyapunov functions $\psi_{i}(x)$
each satisfying
\begin{align}
\psi_{i}(x)=\alpha\dot{\psi}_{i}(x)\indent\forall x\in\mathds{R}^{n}.
\end{align}
\end{proposition}
\begin{proof}
This is proved for linear systems, by constructing the complete abstraction.

Consider a linear differential equation
\begin{align}
\begin{bmatrix}\dot{x}_{1}\\\dot{x}_{2}\end{bmatrix}=
\begin{bmatrix}\lambda_{1}I_{1}&0\\0&\lambda_{2}I_{2}\end{bmatrix}\begin{bmatrix}x_{1}\\x_{2}\end{bmatrix}\label{eqn:diag_lambda_system}
\end{align}
where $I_{1}$, $I_{2}$ are identity matrices and $\lambda_{1}<0$ and $\lambda_{2}>0$.

The stable and unstable subspaces of \eqref{eqn:diag_lambda_system} are orthogonal and can be treated separately.
This system is divided into a stable space described by $x_{1}$ and an unstable space described by $x_{2}$.
For $i\in\{1,2\}$ let $\psi_{i}(x_{i})=x_{i}^{{\rm T}}P_{i}x_{i}$ be a quadratic Lyapunov function. Then its derivative is $\dot{\psi}(x_{i})=x_{i}^{{\rm T}}Q_{i}x_{i}$, where
\begin{align}
2\lambda_{i} P_{i}=Q_{i}\indent\tn{for }i=1,2.
\end{align}

This implies that any quadratic Lyapunov function satisfies Proposition~\ref{prop:ly_tight} and hence generates a complete abstraction.

Since hyperbolic linear systems are topologically conjugate if and only if they have the same index \cite{hirsch}. There is a homeomorphism $h:\mathds{R}^{n}\rightarrow\mathds{R}^{n}$ such that any hyperbolic linear system is topologically conjugate of \eqref{eqn:diag_lambda_system}, by choosing $I_{1}$ and $I_{2}$ appropriately. Note that $h$ is a diffeomorphism on $\mathds{R}^{n}\backslash\{0\}$.

This implies that there exists a complete abstraction of every hyperbolic linear system.
\end{proof}

\section{Conclusion}\label{sec:conclusion}
In this report proofs associated with a method for abstracting hyperbolic dynamical systems by timed automata have been presented. The method is based on partitioning the state space of the dynamical systems by set-differences of invariant sets.

To enable both verification and falsification of safety properties for the considered system based on the abstraction, conditions for soundness, completeness, and refinability have been set up. Furthermore, it is shown that the abstraction can be obtained as a parallel composition of multiple timed automata under certain conditions.

Finally, it is shown that there exist sound and refinable abstractions for hyperbolic Morse-Smale systems. Additionally, it is shown that there exist complete and refinable abstractions for any hyperbolic linear systems.

\bibliographystyle{IEEEtran}
\bibliography{bibliography}

\end{document}